\RequirePackage{amsmath}
\documentclass[runningheads]{llncs}
\usepackage{amssymb}
\usepackage{verbatim}
\usepackage{gastex}
\usepackage{algorithmicx,algorithm,algpseudocode}
\usepackage{pstricks,pstricks-add,pst-node}
\usepackage{array,graphicx}
\usepackage{xcolor}
\usepackage[lowtilde]{url}

\usepackage[utf8]{inputenc}
\usepackage[T1]{fontenc}
\DeclareSymbolFont{rsfscript}{OMS}{rsfs}{m}{n}
\DeclareSymbolFontAlphabet{\mathrsfs}{rsfscript}

\begin{document}
\title{Synchronizing Automata\\with Extremal Properties}
\author{Andrzej Kisielewicz\thanks{Supported in part by Polish MNiSZW grant IP 2012 052272.}
\and Marek Szyku{\l}a\thanks{Supported in part by Polish NCN grant DEC-2013/09/N/ST6/01194.}}
\institute{Department of Mathematics and Computer Science, University of Wroc{\l}aw
\email{andrzej.kisielewicz@math.uni.wroc.pl,\ msz@cs.uni.wroc.pl}}
\maketitle

\begin{abstract}
We present a few classes of synchronizing automata exhibiting certain extremal properties with regard to synchronization. 
The first is a series of automata with subsets whose shortest extending words are of length $\varTheta(n^2)$, where $n$ is the number of states of the automaton.
This disproves a conjecture that every subset in a strongly connected synchronizing automaton is $cn$-extendable, for some constant $c$, and in particular, shows that the cubic upper bound on the length of the shortest reset words cannot be improved generally by means of the extension method.  
A detailed analysis shows that the automata in the series have subsets that require words as long as $n^2/4+O(n)$ in order to be extended by at least one element.

We also discuss possible relaxations of the conjecture, and propose the image-extension conjecture, which would lead to a quadratic upper bound on the length of the shortest reset words. In this regard we present another class of automata, which turn out to be counterexamples to a key claim in a recent attempt to improve the Pin-Frankl bound for reset words.

Finally, we present two new series of slowly irreducibly synchronizing automata over a ternary alphabet, whose lengths of the shortest reset words are $n^2-3n+3$ and $n^2-3n+2$, respectively. These are the first examples of such series of automata for alphabets of size larger than two.
\end{abstract}


\section{Introduction}

In this paper we deal with \emph{deterministic finite (semi)automata} (\emph{DFA}) $\mathrsfs{A} = (Q, \Sigma, \delta)$, where $Q$ is a non-empty
\emph{set of states}, $\Sigma$ is a non-empty \emph{alphabet}, and $\delta\colon Q \times \Sigma \mapsto Q$ is the complete
\emph{transition function}. We extend $\delta$ to $Q \times \Sigma^*$ and $2^Q \times \Sigma^*$ in a natural way.
The image $\delta(S,w)$ is denoted shortly by $Sw$, and the preimage $\delta^{-1}(S,w) = \{q \in Q\mid qw \in S\}$ is denoted by $Sw^{-1}$.
Throughout the paper, by $n$ we denote the cardinality $|Q|$.
The \emph{rank} of a word $w \in \Sigma^*$ is the cardinality $|Qw|$. A word $w$ of rank 1 is called a \emph{synchronizing word}. An automaton for which there exists a synchronizing word is called \emph{synchronizing}.
The \emph{reset length} is the length of the shortest reset words.
A word $w$ \emph{compresses} a subset $S \subseteq Q$ if $|Sw| < |S|$, and a subset that admits such a word is called \emph{compressible}.

The famous \v{C}ern\'y conjecture states that every synchronizing automaton $\mathcal{A}$ with $n$ states has a reset word of length $\leq (n-1)^2$. This conjecture was formulated by \v{C}ern\'y in~1964 \cite{Cerny1964}, and is considered the longest-standing open problem in combinatorial theory of finite automata. So far, the conjecture has been proved only for a few special classes of automata, and the best general upper bound proven is $\frac{n^3-n}{6}-1$ ($n\ge 4)$. It is known that it is enough to prove the conjecture for strongly connected automata. We refer to~\cite{KariVolkov2013Handbook,Volkov2008Survey} for excellent surveys discussing recent results.

Many partial results towards the solution of the \v{C}ern\'{y} conjecture use the so-called \emph{extension method}\index{extension method} (e.g.~\cite{BBP2011QuadraticUpperBoundInOneCluster,BerlinkovSzykula2015AlgebraicSynchronizationCriterion,Berlinkov2013QuasiEulerianOneCluster,Dubuc1998,Kari2003Eulerian,Rystsov1995QuastioptimalBound,Steinberg2011AveragingTrick,Steinberg2011OneClusterPrime}).
The essence of this method is to look for a sequence of (short) words $w_1,\ldots,w_{n-1}$ such that $S_1=\{q\}$ for some $q\in Q$, $S_{k+1}=S_k w_k^{-1}$,  and $|S_{k+1}| > |S_k|$ for $k\ge 1$. Then, necessarily, $S_n=Q$, and the word $w_{n-1}\ldots w_2w_1$ is synchronizing. (It may happen that a shorter sequence of words does the job.) In other words, we try to extend subsets of $Q$, starting from a singleton, rather than to compress subsets starting from $Q$. In connection with this method we say that $S\subseteq Q$ is \emph{$m$-extendable} if there is a word $w$ of length at most $m$ with $|Sw^{-1}| > |S|$. An automaton $\mathcal{A}$ is $m$-extendable, if each subset of its states is $m$-extendable. For example, Kari \cite{Kari2003Eulerian} proved that automata whose underlying digraph is Eulerian are $n$-extendable. A simple calculation shows that such automata satisfy the \v{C}ern\'{y} conjecture.

Unfortunately, it is not true, in general, that each strongly connected synchronizing automaton is $n$-extendable. In \cite{Berlinkov2011OnConjectureCarpiDAlessandro} Berlinkov has constructed a series of strongly connected synchronizing automata with subsets for which the shortest extending words are of length $2n-3$. Thus they are not $cn$-extendable for any constant $c < 2$.
However, it remained open whether each synchronizing automaton is $2n$-extendable. The question was also posed in~\cite{KariVolkov2013Handbook}.
Its importance comes from the fact, that it would imply a quadratic upper bound for the length of the shortest reset words, and many proofs of the \v{C}ern\'{y} conjecture for special cases rely on the fact that subsets are extendable by short words.

In Section~\ref{sec:quadratically_extendable_subsets} we present a series $\mathrsfs{A}_{2m-1}$ of strongly connected synchronizing automata with $n=2m-1$ states that are not $cn$-extendable for any constant $c$.
We show that they have subsets whose shortest extending words have length $n^2/4+O(n)$.
It follows that, in general, the extension method cannot be used to improve the cubic upper bound.
On the other hand, we propose a possible weakening of the condition of $m$-extendability, which would imply quadratic upper bounds. To prove some lower bound on the constant in our conjecture, a series of automata $\mathrsfs{B}_{2m}$ is presented, which turn out to be also connected with a recent attempt to improve the Pin-Frankl bound \cite{Tr2011ModifyingUpperBound}. It provides counterexamples of arbitrary order to the false statement in~\cite[Lemma 3]{Tr2011ModifyingUpperBound}, in addition to the single counterexample on 4 states constructed in~\cite{GJT2014ANoteOnARecentAttempt}.

Automata with reset lengths close to the \v{C}ern\'y bound $(n-1)^2$ are referred to as \emph{slowly synchronizing}.
Such extremal series and particular examples of automata are of special interest, and there are only a few of them known \cite{AGV2010,AGV2013,AVZ2006,GusevPribavkina2014ResetThresholdsOfAutomataWithTwoCycleLengths,Roman2008ANote,Tr2006Trends}.
The first known such series is the famous \v{C}ern\'y series \cite{Cerny1964}, reaching the \v{C}ern\'y bound $(n-1)^2$.
The authors of~\cite{AGV2010} found 10 infinite series of extremal automata other than the \v{C}ern\'y series over binary alphabets.
Yet, until now, no non-trivial series of synchronizing automata with reset length $n^2+O(n)$ have been found for larger alphabets.

Of course, we are not interested in examples that can be obtained just by adding arbitrary letter to binary slowly synchronizing automata. 
We call an automaton \emph{irreducibly synchronizing}, if removing any letter yields a non-synchronizing automaton.
In Section~\ref{sec:slowly_synchronizing_3} we present two such series over a ternary alphabet, with the reset lengths $n^2-3n+3$ and $n^2-3n+2$, respectively.


\section{A Series with Quadratically Extendable Subsets}\label{sec:quadratically_extendable_subsets}

For $m \ge 3$, let $n = 2m-1$. Let $\mathrsfs{A}_{2m-1} = \langle Q_{2m-1}, \{a,b\}, \delta_{2m-1} \rangle$ be the automaton shown in Figure~\ref{fig:backward_hard_odd}. $Q_{2m-1} = \{q_1,\ldots,q_{2m-1}\}$ and $\delta_{2m-1}$ is defined as follows:
{\small
$$\delta_{2m-1}(q_i,a) = \begin{cases}
q_{1}, & \text{ if $i = m$,} \\
q_{m+1}, & \text{ if $i = 2m-1$,} \\
q_{i+1}, & \text{ otherwise,} \\
\end{cases}\quad
\delta_{2m-1}(q_i,b) = \begin{cases}
q_{i}, & \text{ if $1 \le i \le m-1$,} \\
q_{2m-1}, & \text{ if $i = m$,} \\
q_{m}, & \text{ if $i = 2m-1$,} \\
q_{i-m}, & \text{ otherwise.} \\
\end{cases}$$}

\begin{figure}
\unitlength 26pt
\begin{center}\begin{picture}(8,5)(0,0)
\gasset{Nh=1,Nw=1,Nmr=1,ELdist=0.25,loopdiam=0.5}
\node(vm)(8,2){$q_m$}
\node(v1)(0,2){$q_1$}
\node[Nframe=n](vdots1)(2,2){$\dots$}
\node(vm-2)(4,2){$q_{m-2}$}
\node(vm-1)(6,2){$q_{m-1}$}
\drawedge[curvedepth=2,ELdist=-0.25](vm,v1){$a$}
\drawedge(v1,vdots1){$a$}
\drawedge(vdots1,vm-2){$a$}
\drawedge(vm-2,vm-1){$a$}
\drawedge(vm-1,vm){$a$}
\node(v2m-1)(8,4){$q_{2m-1}$}
\node(vm+1)(0,4){$q_{m+1}$}
\node[Nframe=n](vdots2)(2,4){$\dots$}
\node(v2m-2)(4,4){$q_{2m-2}$}
\drawedge[curvedepth=-1.2,ELdist=-0.25](v2m-1,vm+1){$a$}
\drawedge(vm+1,vdots2){$a$}
\drawedge(vdots2,v2m-2){$a$}
\drawedge(v2m-2,v2m-1){$a$}
\drawloop[loopangle=-90,ELdist=0.1](v1){$b$}
\drawloop[loopangle=-90,ELdist=0.1](vm-2){$b$}
\drawloop[loopangle=-90,ELdist=0.1](vm-1){$b$}
\drawedge[curvedepth=0.2](v2m-1,vm){$b$}
\drawedge[curvedepth=0.2](vm,v2m-1){$b$}
\drawedge(vm+1,v1){$b$}
\drawedge(v2m-2,vm-2){$b$}
\end{picture}\end{center}
\caption{The automaton $\mathrsfs{A}_{2m-1}$.}\label{fig:backward_hard_odd}
\end{figure}
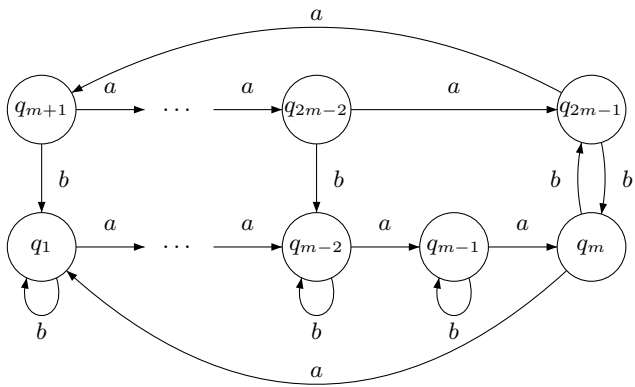

It is easily verified that $\mathrsfs{A}_{2m-1}$ is strongly connected. Let us define $Q_U = \{q_{m+1},\ldots,q_{2m-1}\}$ and $Q_D = \{q_1,\ldots,q_m\}$.
First we show that $\mathrsfs{A}_{2m-1}$ is synchronizing by a word of length $2m^2-2m+2 = (n^2+3)/2$.

\begin{proposition}\label{pro:extend_synchronizing}
The word $b a b a^m b (a^{m-1} b a^m b)^{m-2}$ synchronizes $\mathrsfs{A}_{2m-1}$ to the state $q_1$.
\end{proposition}
\begin{proof}
First observe that $Q b a b a^m b = Q_D \setminus \{q_m\}$.
It suffices to ensure that if $S = \{q_1,\ldots,q_i\}$ for $i=2,\ldots,m-1$, then $S a^{m-1} b a^m b = \{q_1,\ldots,q_{i-1}\}$.
Indeed, $S a^{m-1} = \{q_1,\ldots,q_{i-1},q_m\}$; then $q_m b a^m b = q_1$ and $\{q_1,\ldots,q_{i-1}\}$ is mapped by the action of $b a^m b$ to itself.
\qed
\end{proof}


The following result allows us to refute the hypothesis that there exists a constant $c>0$ such that each subset in a synchronizing automaton is $cn$-extendable. We first prove a weaker estimation of order $\varTheta(n^2)$, because it also allows to refute the hypothesis, while it has a simpler proof.

For any $S \subseteq Q$ we say that $q_i \in S \cap Q_U$ is \emph{covered} (in $S$) if $q_i b \in S$.

\begin{lemma}\label{lem:extend_S}
Let $S \subsetneq Q$ be a non-empty set, whose states from $S \cap Q_U$ are all covered in $S$.
Then a shortest word $w$ such that $|Sw^{-1} \cap Q_D| > |S|$ has length at least $m$.
\end{lemma}
\begin{proof}
Let $w = a_k \dots a_1$ be a shortest word of length $k$ such that $|Sw^{-1} \cap Q_D| > |S|$. Let $S_i = S a^{-1}_1 \dots a^{-1}_i$ for $i=0,\ldots,k$. So, $S_{0} = S$, and $S_k = Sw^{-1}$.

Since the length of $w$ is assumed to be the least possible, the action of $a_k$ cannot be a permutation, and therefore $a_k = b$. Moreover, since the transformation induced by $b$ maps only one state from $Q_D$ to $Q_U$, namely $q_m b = q_{2m-1}$, $S_{k-1}$ must contain $q_{2m-1}$, and must not contain $q_{m}$. Thus, $q_{2m-1}$ is not covered in $S_{k-1}$.

For any $T \subseteq Q$, $Tb^{-2} \cap Q_D \subseteq T \cap Q_D$. It follows that, if $a_{k-1} = b$, then $|S_{k-2} \cap Q_D| = |S_k \cap Q_D|$, which contradicts the assumption about the shortest length. Hence $a_{k-1} = a$.
Now, it follows that $S_{k-2}$ contains $q_{m+1}$, which is uncovered. Consequently, $a_{k-2} = a$, as any preimage under $b$ does not contain an uncovered state other than $q_{2m-1}$. Similarly, $S_{k-3}$ contains uncovered state $q_{m+2}$ (provided $m>3$). Repeating this argument $2m-1-(m+2) = m-3$ times we have that $a_i = a$ for $k-1 \ge i \ge k-m+1$, and $|w| = k \ge m$ as required.
\qed
\end{proof}

\begin{theorem}\label{thm:ext_easy}
The shortest words extending the subset $Q_U$ have length at least $2 + m\lceil(m-3)/2\rceil$.
\end{theorem}
\begin{proof}
Let $u = a_k \ldots a_1$ be a shortest word such that $|Q_U u^{-1}| > |Q_U|$.
Obviously, $u$ ends with $b$, that is, $a_1 = b$, since $Q_U a^{-1} = Q_U$.
Also, as in the proof of Lemma~\ref{lem:extend_S}, $a_k=b$. Moreover, $Q_Ub^{-1} = \{q_m\}$.

For a subset $X \subseteq Q$ we define $d(X) = |X \cap Q_D|$.
Observe that $|Xb^{-1}| \le 2d(X)+1$.

Let $$T = \{q_m\} a^{-1}_2 \ldots a^{-1}_{k-1}.$$
Since $|Q_U u^{-1}| = |T a^{-1}_k| \ge m$, and $a_k = b$, we have that $m \le 2d(T)+1$, and so, $d(T) \ge \lceil(m-1)/2\rceil$.

Thus, we have $u = bvb$, where $v = a_{k-1} \dots a_2$, and $\{q_m\} v^{-1} = T$.
We define by induction the consecutive factors of $v$ as follows: 
$v=v_\ell v_{\ell-1} \ldots v_1$, $T_0=\{q_m\}$, and for each $i=1,2,\ldots,\ell$, $T_i = \{q_m\} v^{-1}_1 \ldots v^{-1}_i$, and $v_i$ is the shortest factor of $v$ such that $d(T_i) > d(T_{i-1})$.
Since they are the shortest words whose inverse action increases the number $d(T_{i-1})$, each $v_i$ begins with $b$, and the increase is by one: $d(T_i) = d(T_{i-1})+1$.
Moreover, for each $i$, all the states in $T_i \cap Q_U$ are covered in $T_i$.

By applying Lemma~\ref{lem:extend_S} for each $T_i$, we obtain $|v_i| \ge m$.
Since $d(q_m) = 1$ and $d(T) \ge \lceil(m-1)/2\rceil$, we have that $\ell \ge \lceil(m-3)/2\rceil$.
Thus $|u| \ge 2 + m\lceil(m-3)/2\rceil$ as required.
\qed
\end{proof}

The lower bound in the theorem above is rather rough. While it suffices to refute the conjecture in question, it is natural to ask for a better estimation. In this connection we have the following

\begin{theorem}\label{thm:ext_hard}
The shortest words extending the subset $Q_U$ have length $m^2+O(m)$.
Any non-empty proper subset $S \subset Q$ can be extended by a word of length at most $m^2+O(m)$.
\end{theorem}

The proof of this result is much longer and involved than the proof of Theorem~\ref{thm:ext_easy}, and therefore it is not reproduced here. We describe only its main idea.

The beginning of the proof is the same as in the proof of the previous theorem. We make use of the fact that to extend the set $\{q_m\}$ to a set of cardinality $2d$ (for some $d$) one needs to obtain first a set $\{q_m\}w^{-1}=X=U\cup D$ with $|D| \geq d$, $D \subseteq Q_D$, $U \subseteq Q_U$. Next, we show that to achieve this aim the most effective (using a shortest word) is the greedy procedure described below.   

First we use the word
$w=b a^{2m-3} b a^2$.
It is easy to check that $\{q_m\}w^{-1} = \{q_1,q_m\}$. Then we use repeatedly $u = b a^{2m-1}$, obtaining the sets
$\{q_1,\ldots,q_i,q_m\}$ for $i=2,3,\ldots,d$.
Using the greedy strategy we can find a word extending $Q_U$ of length $m^2 - 3m/2 + 4$, in case of even $m$, and $m^2 - m +2$, in case of odd $m$.
This yields easily the second statement. The main difficulty is to demonstrate that the greedy strategy is the most efficient one. It requires to consider other possible strategies and to introduce special terminology to handle this.


We need to remark, that the fact that the greedy strategy is the most efficient way to extend the subset $\{q_m\}$ by a chosen number of elements does not mean that it is the most efficient way to extend the subset $Q_U$. The problem is that after obtaining the required number of states in $D$, which in case of even $m$ is exactly $m/2$,
we may still need a translation of $D$, so that the required number of states is in $D \cap \{q_1,\ldots,q_{m-2}\}$
(and applying subsequently $b$ doubles the number of states).
We have checked that sometimes different strategies to extend $Q_U$ lead to words of the same length as the one obtained by the greedy strategy. It may happen that there are cases when a different strategy is slightly more efficient (within $O(m)$ summand).

Note also, that if the \v{C}ern\'{y} conjecture is true, then every synchronizing automaton is $(n-1)^2$-extendable, thus asymptotically no better bound is possible.
However, we also believe that our series represents the worst possible case up to $O(n)$, that is, we conjecture that every strongly connected synchronizing automaton is $(n^2/4+O(n))$-extendable.


It is possible to define a similar series for all even $n = 2m$ (Figure~\ref{fig:backward_hard_even}). The counterparts of Theorems~\ref{thm:ext_easy} and~\ref{thm:ext_hard} can be proved in a similar way.

\begin{figure}
\unitlength 26pt
\begin{center}\begin{picture}(8,5)(0,0.5)
\gasset{Nh=1,Nw=1,Nmr=1,ELdist=0.25,loopdiam=0.5}
\node(vm)(8,2){$q_m$}
\node(v1)(0,2){$q_1$}
\node[Nframe=n](vdots1)(2,2){$\dots$}
\node(vm-2)(4,2){$q_{m-2}$}
\node(vm-1)(6,2){$q_{m-1}$}
\drawedge[curvedepth=2,ELdist=-0.25](vm,v1){$a$}
\drawedge(v1,vdots1){$a$}
\drawedge(vdots1,vm-2){$a$}
\drawedge(vm-2,vm-1){$a$}
\drawedge(vm-1,vm){$a$}
\node(v2m)(8,4){$q_{2m}$}
\node(vm+1)(0,4){$q_{m+1}$}
\node[Nframe=n](vdots2)(2,4){$\dots$}
\node(v2m-2)(4,4){$q_{2m-2}$}
\node(v2m-1)(6,4){$q_{2m-1}$}
\drawedge[curvedepth=-1.2,ELdist=-0.25](v2m,vm+1){$a$}
\drawedge(vm+1,vdots2){$a$}
\drawedge(vdots2,v2m-2){$a$}
\drawedge(v2m-2,v2m-1){$a$}
\drawedge(v2m-1,v2m){$a$}
\drawloop[loopangle=-90,ELdist=0.1](v1){$b$}
\drawloop[loopangle=-90,ELdist=0.1](vm-2){$b$}
\drawloop[loopangle=-90,ELdist=0.1](vm-1){$b$}
\drawedge[curvedepth=0.2](v2m,vm){$b$}
\drawedge[curvedepth=0.2](vm,v2m){$b$}
\drawedge(vm+1,v1){$b$}
\drawedge(v2m-2,vm-2){$b$}
\drawedge[ELdist=-0.3](v2m-1,vm){$b$}
\end{picture}\end{center}
\caption{The automaton $\mathrsfs{A}_{2m}$.}\label{fig:backward_hard_even}
\end{figure}
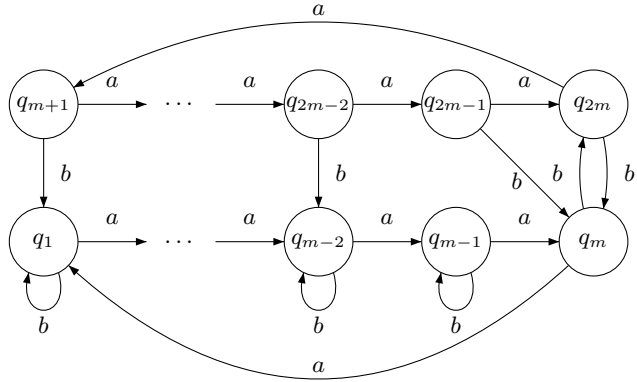

\section{Relaxing the Extension Property}

In view of the examples of subsets that are not extendable by words of length $n$ \cite{Berlinkov2011OnConjectureCarpiDAlessandro}, some efforts were made to relax the extension conjecture (see also the discussion in~\cite{KariVolkov2013Handbook}).
For example, one could suppose that extending an easily (linearly) extendable subset cannot lead to a difficultly (quadratically) extendable subset (this  is so  ]for $\mathrsfs{A}_{2m-1}$, and in general, would allow us to prove a quadratic bound for the length of the shortest reset words).
Berlinkov proposed\footnote{First Russian-Finnish Symposium on Discrete Mathematics \mbox{(RuFiDiM 2011)}} the following relaxed \emph{conservative extension conjecture}:
There exists a constant $c$ such that if a subset $S \subset Q$ can be extended to a subset $T \subset Q$ by a word $v$ of length at most $cn$ ($T = Sv^{-1} \supset S$), then also $T$ can be extended by a word $u$ of length at most $cn$ ($T u^{-1} \supset T$).

This holds true for the series from~\cite{Berlinkov2011OnConjectureCarpiDAlessandro} and also for our series $\mathrsfs{A}_{2m-1}$.
However, we  construct a counterexample by modifying our series; see Figure~\ref{fig:conservative_backward_hard_even}.
Observe that $S = \{q_{m+1},\ldots,q_{2m-1}\}$ is extended by $a$ to $T = S \cup \{q_{2m}\}$. But then we must apply $b^{-1}$, resulting in $\{q_m\}$. Again, the arguments from  the proof of Theorem~\ref{thm:ext_easy}  hold and we need a word of length $\varOmega(n^2)$ to extend it to a subset larger than $S$.
So in fact, this is a counterexample for the stronger statement with $v$ of length 1 and without assuming that $T u^{-1}$ contains $T$.

\begin{figure}
\unitlength 26pt
\begin{center}\begin{picture}(9,6)(0,0.5)
\gasset{Nh=1,Nw=1,Nmr=1,ELdist=0.25,loopdiam=0.5}
\node(vm)(8,2){$q_m$}
\node(v1)(0,2){$q_1$}
\node[Nframe=n](vdots1)(2,2){$\dots$}
\node(vm-2)(4,2){$q_{m-2}$}
\node(vm-1)(6,2){$q_{m-1}$}
\drawedge[curvedepth=2,ELdist=-0.25](vm,v1){$a$}
\drawedge(v1,vdots1){$a$}
\drawedge(vdots1,vm-2){$a$}
\drawedge(vm-2,vm-1){$a$}
\drawedge(vm-1,vm){$a$}
\node(v2m-1)(6,4){$q_{2m-1}$}
\node(vm+1)(0,4){$q_{m+1}$}
\node[Nframe=n](vdots2)(2,4){$\dots$}
\node(v2m-2)(4,4){$q_{2m-2}$}
\node(v2m)(8,4){$q_{2m}$}
\drawedge[curvedepth=-1.1,ELdist=-0.25](v2m-1,vm+1){$a$}
\drawedge(vm+1,vdots2){$a$}
\drawedge(vdots2,v2m-2){$a$}
\drawedge(v2m-2,v2m-1){$a$}
\drawloop[loopangle=-90,ELdist=0.1](v1){$b$}
\drawloop[loopangle=-90,ELdist=0.1](vm-2){$b$}
\drawloop[loopangle=-90,ELdist=0.1](vm-1){$b$}
\drawedge[curvedepth=0.2](v2m,vm){$b$}
\drawedge[curvedepth=0.2](vm,v2m){$b$}
\drawedge[curvedepth=-1.6,ELdist=-0.4](v2m,vm+1){$a$}
\drawedge(v2m-1,vm-1){$b$}
\drawedge(vm+1,v1){$b$}
\drawedge(v2m-2,vm-2){$b$}
\end{picture}\end{center}
\caption{An automaton with a quadratically extendable preimage.}\label{fig:conservative_backward_hard_even}
\end{figure}
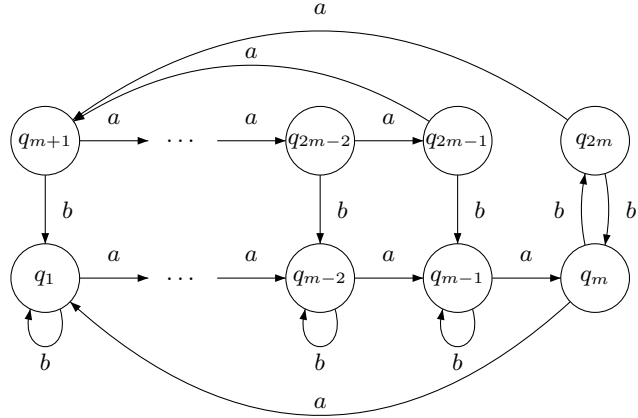

\subsection{Image Extending Conjecture}\label{subsec:image_extending}

One may observe that the subsets $Q_{2m-1}$ requiring long extending words are images of $Q_{2m-1}$ under the action of no word.
Hence, they cannot appear as intermediate subsets in the process of applying the consecutive letters of a synchronizing word. From the point of view of the extension method, if we consider prefixes $w_1 w_2 \dots w_i$ with $|Q w_1 \dots w_i| > |Q w_1 \dots w_{i+1}|$, we only need to find extending words for the images $Q w_1 \dots w_{i+1}$.
Thus, to apply iteratively extending words starting from a singleton, at each step, the resulted preimage must be also an image, or at least it needs to contain an image of the larger cardinality than the set in question. 

This restriction leads to the following weaker conjecture that we believe is true.

\begin{conjecture}[Image-Extension]\label{con:imgext}
There exists a constant $c$ such that, 
for every strongly connected synchronizing automaton $\mathrsfs{A} = (Q, \Sigma, \delta)$
and every non-empty $S \subset Q$ with $Qw=S$ for some $w \in \Sigma^*$, there exists a word $u$ of length at most $cn$ and a subset $T\subset Q$ such that $T \subseteq Su^{-1}$, 
$Qv=T$ for some $v \in \Sigma^*$, and $|T| > |S|$.
\end{conjecture}  

This conjecture implies a quadratic bound on the length of the shortest synchronizing word, yet we show that it cannot be a tool to prove the \v{C}ern\'y conjecture because of the following

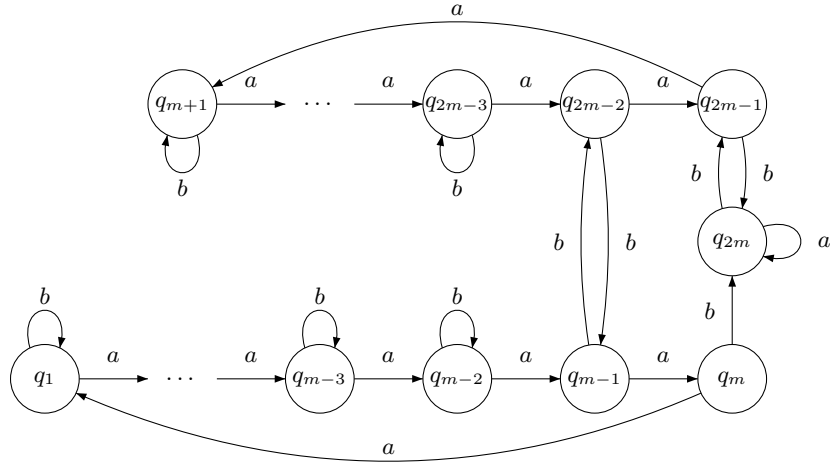
\begin{figure}
{\unitlength 26pt
\begin{center}\begin{picture}(10,7)(0,0.5)
\gasset{Nh=1,Nw=1,Nmr=1,ELdist=0.25,loopdiam=0.5}
\node(v1)(0,2){$q_1$}
\node[Nframe=n](vdots1)(2,2){$\dots$}
\node(vm-3)(4,2){$q_{m-3}$}
\node(vm-2)(6,2){$q_{m-2}$}
\node(vm-1)(8,2){$q_{m-1}$}
\node(vm)(10,2){$q_m$}
\drawedge[curvedepth=1.2,ELdist=-0.25](vm,v1){$a$}
\drawedge(v1,vdots1){$a$}
\drawedge(vdots1,vm-3){$a$}
\drawedge(vm-3,vm-2){$a$}
\drawedge(vm-2,vm-1){$a$}
\drawedge(vm-1,vm){$a$}
\drawloop[loopangle=90,ELdist=0.1](v1){$b$}
\drawloop[loopangle=90,ELdist=0.1](vm-3){$b$}
\drawloop[loopangle=90,ELdist=0.1](vm-2){$b$}
\node(vm+1)(2,6){$q_{m+1}$}
\node[Nframe=n](vdots2)(4,6){$\dots$}
\node(v2m-3)(6,6){$q_{2m-3}$}
\node(v2m-2)(8,6){$q_{2m-2}$}
\node(v2m-1)(10,6){$q_{2m-1}$}
\node(v2m)(10,4){$q_{2m}$}
\drawloop[loopangle=-90,ELdist=0.1](vm+1){$b$}
\drawloop[loopangle=-90,ELdist=0.1](v2m-3){$b$}
\drawedge[curvedepth=-1.2,ELdist=-0.25](v2m-1,vm+1){$a$}
\drawedge(vm+1,vdots2){$a$}
\drawedge(vdots2,v2m-3){$a$}
\drawedge(v2m-3,v2m-2){$a$}
\drawedge(v2m-2,v2m-1){$a$}
\drawloop[loopangle=0](v2m){$a$}
\drawedge[curvedepth=0.2](v2m-2,vm-1){$b$}
\drawedge[curvedepth=0.2](vm-1,v2m-2){$b$}
\drawedge(vm,v2m){$b$}
\drawedge[curvedepth=0.2](v2m,v2m-1){$b$}
\drawedge[curvedepth=0.2](v2m-1,v2m){$b$}
\end{picture}\end{center}
\caption{The automaton $\mathrsfs{B}_{2m}$.}\label{fig:weakext_3/2}
}\end{figure}

\begin{proposition}
The constant $c$ in Conjecture~\ref{con:imgext} must be at least $3/2$.
\end{proposition}
\begin{proof}
For $n = 2m \ge 8$, consider the automaton $\mathrsfs{B}_{2m}$ from Figure~\ref{fig:weakext_3/2}.
Clearly $\mathrsfs{B}_{2m}$ is strongly connected. To show that it is synchronizing, it is sufficient to prove that every pair of states can be compressed.
Let $\{p,q\}$ be a pair of distinct states.
Suppose that $p$ and $q$ lie in different cycles of $a$; without loss of generality, $p \in \{q_{m+1},\ldots,q_{2m-1}\}$ and $q \in \{q_1,\ldots,q_m\}$.
Since the lengths of the cycles of $a$ are relatively prime, by a word of the form $a^i$ we can map $p,q$ to any other pair of states with $pa^i \in \{q_{m+1},\ldots,q_{2m-1}\}$ and $qa^i \in \{q_1,\ldots,q_m\}$. In particular, we can map them to $q_{2m-1}$ and $q_{m}$. Then $b$ compresses this pair.
If $p$ and $q$ lie in the upper cycle $(q_{m+1},\ldots,q_{2m-1})$ of $a$, then by a word of the form $a^i$ we can map one of the states to $q_{2m-2}$ so that the second state is not mapped to $q_{2m-1}$. Then using $b$ results in a pair of states in different cycles of $a$, and we repeat the argument above.
Similarly, if $p$ and $q$ lie in the lower cycle $(q_1,\ldots,q_m)$, then we can map one of them to $q_{m-1}$ and the second one to a state other than $q_m$, and again apply $b$.
Finally, if $p = q_{2m}$ then either using $b$ or $ab$ (depending on $q$) results in a pair without $q_{2m}$.

Consider now $S=\{q_{m-3},q_{m-2}\}$.
Since $b$ can reduce the size of the subset only by 1, and $a$ is a permutation, there exists a word of rank 2.
Then we can map the resulted pair $\{p,q\}$ onto $S$ by the arguments above. Thus $S$ is an image of $Q$ under the action of some word.

We show that the shortest words extending $S$ have length $\frac{3n}{2}-1$.
Let $u$ be a shortest extending word of $S$. To simplify notation in the remaining part of the proof, we consider the reversed word $w = u^R$, and the inverse actions $a^{-1}$, $b^{-1}$ of both letters.

Obviously $w$ starts with $a^{m-2}$, since applying $b$ earlier does not result in a new set, or results in a proper subset in the case of $\{q_m,q_1\}$.
Then we have $\{q_{m-1},q_m\}$, and we consider the following two cases:
\begin{enumerate}
\item The next letter is $b$.
Here we have $\{q_{2m-2}\}$, and the next part of $w$ is $a^{m-2}$, which results in $\{q_{2m-1}\}$.
Then there must be $b^2$, and we obtain $\{q_{2m-1},q_m\}$.
The length of $w$ considered so far is $(m-2)+1+(m-2)+2 = 2m-1$.
\item The next letter is $a$.
Here we have $\{q_{m-2},q_{m-1}\}$, and the next letter is $b$.
The next part of $w$ is $a^{m-2}$, which results in $\{q_{2m-1},q_m\}$.
The length of $w$ considered so far is $(m-2)+1+1+(m-2) = 2m-2$.
\end{enumerate}
Since $w$ is a shortest word, the second case must take place.
The next part of $w$ must be $a^{m-1}$, which maps $\{q_{2m-1},q_m\}$ to $\{q_{2m-1},q_1\}$.
Then $b^2$ is the shortest word extending $\{q_{2m-1},q_1\}$, which results in $\{q_{2m-1},q_m,q_1\}$.
Therefore, the length of $w$ is $2m-2 + (m-1) + 2 = 3m-1 = 3n/2-1$.
\qed
\end{proof}

\subsection{Separating States}\label{subsec:separating_states}

The series of automata $\mathrsfs{B}_{2m}$ in Figure~\ref{fig:weakext_3/2} has another interesting property connected with a recent attempt to improve the
Pin-Frankl bound \cite{Tr2011ModifyingUpperBound}. Lemma~3 in the above mentioned paper, claiming that for every state $q$ there exists a word $w$ of length at most $n$ such that $q \not\in Qw$, turned out to be false. 
In a recent short note \cite{GJT2014ANoteOnARecentAttempt}
the authors present a certain automaton on 4 states and show that it is a counterexample to the statement in~\cite[Lemma 3]{Tr2011ModifyingUpperBound}. In this connection it is worth observing that our series $\mathrsfs{B}_{2m}$ provides counterexamples of arbitrary order.

It is not difficult to prove the following

\begin{proposition}\label{pro:separating_states}
The shortest words $w$ such that $q_{2m} \not\in Qw$ are of length $2m+2 = n+2$.
\end{proposition}

The length $w$ in our proposition exceeds $n$ only by 2. 
It is interesting whether there are counterexamples requiring still longer words than those of length $n+2$. The open question in \cite{GJT2014ANoteOnARecentAttempt} asks whether there exists a constant $c$ such that, for any finite automaton and its state $q$, there exists a word $w$ of length not greater than $cn$ such that $q \notin Qw$.


\section{Slowly Synchronizing Automata on a Ternary Alphabet}\label{sec:slowly_synchronizing_3}

In this section we present two series of ternary irreducibly slowly synchronizing automata. They turn out to be related to the digraph $W_n$ defined in  \cite{AGV2013}. They are of interest since so far only such series over a binary alphabet have been known. 

Let $Q_n = \{q_1,\ldots,q_n\}$, $n\ge 3$, and $\Sigma = \{a,b,c\}$.
Let $\mathrsfs{M}_n = \langle Q_n, \Sigma, \delta_n \rangle$ and $\mathrsfs{M}'_n = \langle Q_n, \Sigma, \delta'_n \rangle$ be the automata shown in~Figure~\ref{fig:slowly_synchronizing_mn}. The transition functions $\delta_n$ and $\delta'_n$ are defined as follows:
$$\delta(q_i,a) = \begin{cases}
q_{i+1}, & \text{ if $1 \le i \le n-1$,} \\
q_{2}, & \text{ if $i = n$}, \\
\end{cases}\quad\quad
\delta(q_i,c) = \begin{cases}
q_{n}, & \text{ if $i = 1$,} \\
q_{i}, & \text{ if $2 \le i \le n-1$,} \\
q_{1}, & \text{ if $i = n$}, \\
\end{cases}$$
$$\delta(q_i,b) = \begin{cases}
q_{2}, & \text{ if $i = 1$,} \\
q_{i}, & \text{ if $2 \le i \le n-1$,} \\
\end{cases}\quad\quad\;\;
\delta'(q_i,c) = \begin{cases}
q_{i}, & \text{ if $1 \le i \le n-1$,} \\
q_{1}, & \text{ if $i = n$,} \\
\end{cases}$$
\noindent and $\delta'(q_i,a) = \delta(q_i,a), \delta'(q_i,b) = \delta(q_i,b)$, otherwise.

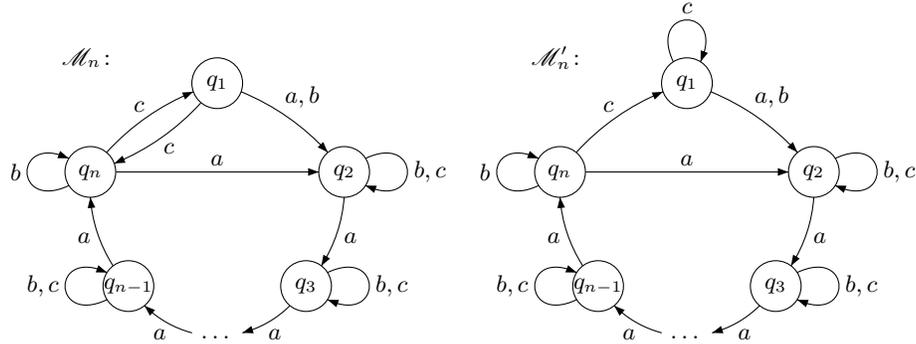
\begin{figure}
\unitlength 4.8pt
\begin{center}\begin{picture}(19,23)(10,5.5)
\gasset{Nh=4,Nw=4,Nmr=2,ELdist=0.5,loopdiam=3}
\node[Nframe=n](name_m1n)(0,26){$\mathrsfs{M}_n\colon$}
\node(vn-1)(3,8){$q_{n-1}$}
\node(vn)(0,17){$q_n$}
\node(v1)(10,24){$q_1$}
\node(v2)(20,17){$q_2$}
\node(v3)(17,8){$q_3$}
\node[Nframe=n](vdots)(10,4){$\dots$}
\drawedge[curvedepth=1](vdots,vn-1){$a$}
\drawedge[curvedepth=1](vn-1,vn){$a$}
\drawedge(vn,v2){$a$}
\drawedge[curvedepth=1](v1,v2){$a,b$}
\drawedge[curvedepth=1](v2,v3){$a$}
\drawedge[curvedepth=1](v3,vdots){$a$}
\drawedge[curvedepth=1](vn,v1){$c$}
\drawedge[curvedepth=1](v1,vn){$c$}
\drawloop[loopangle=0](v2){$b,c$}
\drawloop[loopangle=0](v3){$b,c$}
\drawloop[loopangle=180](vn-1){$b,c$}
\drawloop[loopangle=180](vn){$b$}
\end{picture}\begin{picture}(19,23)(-8,5.5)
\gasset{Nh=4,Nw=4,Nmr=2,ELdist=0.5,loopdiam=3}
\node[Nframe=n](name_m2n)(0,26){$\mathrsfs{M}'_n\colon$}
\node(vn-1)(3,8){$q_{n-1}$}
\node(vn)(0,17){$q_n$}
\node(v1)(10,24){$q_1$}
\node(v2)(20,17){$q_2$}
\node(v3)(17,8){$q_3$}
\node[Nframe=n](vdots)(10,4){$\dots$}
\drawedge[curvedepth=1](vdots,vn-1){$a$}
\drawedge[curvedepth=1](vn-1,vn){$a$}
\drawedge(vn,v2){$a$}
\drawedge[curvedepth=1](v1,v2){$a,b$}
\drawedge[curvedepth=1](v2,v3){$a$}
\drawedge[curvedepth=1](v3,vdots){$a$}
\drawedge[curvedepth=1](vn,v1){$c$}
\drawloop[loopangle=90](v1){$c$}
\drawloop[loopangle=0](v2){$b,c$}
\drawloop[loopangle=0](v3){$b,c$}
\drawloop[loopangle=180](vn-1){$b,c$}
\drawloop[loopangle=180](vn){$b$}
\end{picture}\end{center}
\caption{$\mathrsfs{M}_n$ with reset length $n^2-3n+3$, and $\mathrsfs{M}'_n$ with reset length $n^2-3n+2$.}\label{fig:slowly_synchronizing_mn}
\end{figure}

In determining the reset length of the series we applied a technique, which is alternative to that of \cite{AGV2010,AGV2013}. Our method is based on analyzing the behavior of the inverse BFS algorithm finding the length of the shortest reset words \cite{KKS2015ComputingTheShortestResetWords} and is suitable, in general, for a more mechanical way to establish the reset lengths of concrete automata. Also, it may lead to simpler proofs in case of larger alphabets (when the number of induced automata by combinations of letters is large).
It relies on the fact, that when we are applying the inverse actions of letters starting from all singletons, then the number of new resulted sets is always bounded by a very small constant. In fact, a particular form of this method was first used for the \v{C}ern\'y series \cite{Cerny1964}.
Interestingly, this method also works fine in a very similar way for all slowly synchronizing series defined in~\cite{AGV2010,AGV2013,AVZ2006}, as they all have this property.

For a synchronizing automaton $\mathrsfs{A} = \langle Q, \Sigma, \delta\rangle$ with $n > 1$ states, we define the sequence of families $(L_i)$ of the subsets of $Q$. Let $L_0$ be the family of all the singletons, which are a common end of more than one edge with the same label. We will define $L_i$ inductively for $i \ge 1$. Let $L'_i = \{S a^{-1}\colon S \in L_{i-1}, a \in \Sigma\}$. A set $S \in L'_i$ is called \emph{visited}, if $|S|=1$ or there is $T \in L_j$, $T \supseteq S$ for $j < i$, or there is $T \in L'_i$, $T \supsetneq S$. We define $L_i$ to be the set of all non-visited sets from $L'_i$.

The proof of the following lemma in a more general form can be found in~\cite[Theorem 1]{KKS2015ComputingTheShortestResetWords}.

\begin{lemma}\label{lem:ibfs}
There exists a shortest reset word $w$ such that for any suffix $u$ of $w$ of length $i$, $\{q\}u^{-1} \in L_i$ for some $q \in Q$. The smallest $i$ such that $Q \in L_i$ is the reset length of the automaton.
\end{lemma}

\begin{theorem}\label{thm:slowly3_reset_length}
For $n \ge 3$, the automata $\mathrsfs{M}_n$ and $\mathrsfs{M}'_n$ are irreducibly synchronizing. The first has reset length $n^2-3n+3$, and the second has reset length $n^2-3n+2$.
\end{theorem}
\begin{proof}
One easily verifies that the word $acb(a^{n-2}cb)^{n-3}$ synchronizes $\mathrsfs{M}_n$, and has the length $3+(n-2+2)(n-3) = n^2-3n+3$.
Also the word $cb(a^{n-2}cb)^{n-3}$ synchronizes $\mathrsfs{M}'_n$, and has the length $n^2-3n+2$.

We show that there is no shorter reset word for $\mathrsfs{M}_n$.
By Lemma~\ref{lem:ibfs}, it is sufficient to show what is the smallest $i$ such that $Q_n \in L_i$, which is the length of the shortest reset words. Here only $q_2$ is a common end of more than one edge with the same label, so $L_0 = \{\{q_2\}\}$.

We claim that for each $i$ with $0 \le i \le n-3$, $L_{in} = \{\{q_2,\ldots,q_{2+i}\}\}$, $L_{(i-1)n+4} = \{\{q_{n-2},q_{n-1},q_n,q_1,\ldots,q_{i-1}\}\}$ if $2 \le i \le n-3$, and $L_{(i-1)n+4} = L_4 = \{\{q_{n-2},q_{n-1}\}\}$ if $i=1$.
The proof follows by induction. Clearly for $i=0$ the claim holds.
Consider some $i$, and assume that the claim holds for $i' \le i$, so $L_{in} = \{\{q_2,\ldots,q_{2+i}\}\}$.
We will show the claim for $i+1$, that is, for $L_{(i+1)n}$ and $L_{in+4}$.

The action of $c^{-1}$ for the set from $L_{in}$ results in the same set, so we have $L_{in+1} = \{S_1,T_1\}$, where $S_1=\{q_n,q_1,\ldots,q_{1+i}\}$ was obtained by the action of $a^{-1}$ and $T_1=\{q_1,\ldots,q_{2+i}\}$ by the action of $b^{-1}$.
Consider the sets obtained from $S_1$. Observe that if a set contains both $q_n$ and $q_1$, then only the action of $a^{-1}$ can result in a non-visited set. If $i=0$ then $S_1 a^{-1} = \{q_n,q_1\}a^{-1} = \{q_{n-1}\}$ is a visited singleton. If $i \ge 1$ then $S_2 = S_1 a^{-1} = \{q_{n-1},q_n,q_1,\ldots,q_i\}$. But $S_2 a^{-1}$ is $\{q_{n-2},q_{n-1}\}$ if $i=1$, or $\{q_{n-2},q_{n-1},q_n,q_1,\ldots,q_{i-1}\}$ if $i \ge 2$; so $S_2 a^{-1}$ is visited by the assumption of $L_{(i-1)n+4}$.
Consider the sets obtained from $T_1$. Only $T_2 = T_1 c^{-1} = \{q_n,q_2,\ldots,q_{2+i}\}$ is non-visited. Then let $T_3 = T_2 a^{-1} =\{q_{n-1},q_n,q_1,\ldots,q_{1+i}\}$. For $T_2 b^{-1} = \{q_n,q_1,\ldots,q_{2+i}\}$ observe that in the next step it results either in $T_3$ or in itself. Only the action of $a^{-1}$ applied to $T_3$ results in a non-visited set, so $L_{in+4} = \{\{q_{n-2},q_{n-1}\}\}$ if $i=0$, and $L_{in+4} = \{\{q_{n-2},q_{n-1},q_n,q_1,\ldots,q_i\}\}$ if $i \ge 1$.
Now, if $i=0$ then only the action of $a^{-1}$ results in a non-visited set over the next $n-4$ steps resulting in $L_{n} = \{q_2,q_3\}$.
Similarly, if $i \ge 1$ then by the next $i-1$ steps by the action of $a^{-1}$ we have $\{q_{n-1-i},\ldots,q_n,q_1\}$.
Again, through the next $n-3-i$ steps only the action of $a^{-1}$ results in a non-visited set, and we finally have $L_{(i+1)n} = \{\{q_2,\ldots,q_{3+i}\}\}$.

From the claim it follows that $Q_n$ does not appear in $L$ for $i \le (n-3)n$, and $L_{(n-3)n} = \{\{q_2,\ldots,q_{n-1}\}\}$. Then applying $(acb)^{-1}$ or $(bcb)^{-1}$ results in $Q_n$, and there is no shorter such word as is easily verified. Hence $L_{(n-3)n+3} = \{\{Q_n\}\}$ and so $i = (n-3)n+3 = n^2-3n+3$ is the length such that $Q_n \in L_i$.

The proof for the automaton $\mathrsfs{M}'_n$ follows exactly in the same way, with the following two exceptions: $T_2 = \{q_n,q_1,\ldots,q_{2+i}\}$, and finally we apply $(cb)^{-1}$ to the set from $L_{(n-3)n}$, resulting in $Q_n$.

It remains to show that removing any letter in $\mathrsfs{M}_n$ ($\mathrsfs{M}'_n$) results in a non-synchronizing automaton. Indeed, removing the letter $a$ results in unconnected states $q_3,\ldots,q_{n-1}$.
The only compressible pairs of states are $\{q_1,q_n\}$ under $a$, and $\{q_1,q_2\}$ under $b$.
Observe that it is not possible to map the pair $\{q_{n-1},q_n\}$ to $\{q_1,q_n\}$: Only $a$ maps a state from $Q_n \setminus \{q_1,q_n\}$ to $\{q_1,q_n\}$, and it maps exactly one such state. However both $q_1$ and $q_n$ are mapped by $a$ to $Q_n \setminus \{q_1,q_n\}$. Hence removing the letter $b$ results in a non-synchronizing automaton.
Removing the letter $c$ makes $q_1$ unreachable from the other states, hence no pair can be compressed except $\{q_1,q_n\}$ and $\{q_1,q_2\}$.
\qed
\end{proof} 

It seems that when we admit more letters in the alphabet, then it is more difficult to find any series of irreducible  strongly connected slowly synchronizing automata.
Although there are many binary series with the reset length $n^2+O(n)$ (obtained by modifying the series from \cite{AGV2010,AGV2013}), it is only $\mathrsfs{M}_n$ and $\mathrsfs{M}'_n$ that are known over a ternary alphabet, and no such series is known over a larger alphabet. 
Thus it becomes an interesting problem to find more such series, if they exist at all.

\medskip
\noindent
{\bf Acknowledgment}
We are grateful to Jakub Kowalski and anonymous referees for careful proofreading.

\bibliographystyle{splncs03}

\end{document}